\documentclass{preprint}
  \usepackage{times}%
  \usepackage{amsmath}%
  \usepackage{mathrsfs}
  \usepackage{mhequ}
  \usepackage{mhenvs}
  \usepackage{Symbols}

\newtheorem{assumption}{Assumption}

\def\E{\mathbf{E}}

\def\Lip{{\mathrm{Lip}_{\delta,\beta}}}

\def\span{\mop{span}}
\def\N{\CN}
\def\J{\mathscr{J}}
\def\D{\mathscr{D}}

\let\d\partial

\begin{document} 

  \title{A simple framework to justify linear response theory}
  \author{Martin~Hairer and Andrew~J.~Majda}
  \institute{Courant Institute, New York University
   \\ \email{mhairer@cims.nyu.edu}, \email{jonjon@cims.nyu.edu}}
  \titleindent=0.65cm

  \maketitle
  \thispagestyle{empty}

\begin{abstract}
The use of linear response theory for forced dissipative
stochastic dynamical systems through the fluctuation dissipation
theorem is an attractive way to study climate change systematically among other
applications. Here, a mathematically rigorous justification of linear
response theory for forced dissipative stochastic dynamical systems is
developed. The main results are formulated in an abstract setting and apply to suitable
systems, in finite and infinite dimensions, that are of interest in climate change 
science and other applications. 
\end{abstract}

\section{Introduction}

One of the cornerstones of modern statistical physics is the fluctuation dissipation theorem (FDT) which 
roughly states that for systems of identical particles in statistical equilibrium, the average response to small external
perturbations can be calculated through the knowledge of suitable correlation
functions of the unperturbed statistical system. For some of the many practical applications of the FDT, see 
for example \cite{Kubo,Balescu}. The low frequency response to changes in external forcing for various components of the climate system is 
a central problem of contemporary climate change science. Leith \cite{Leith} suggested that if the
climate system satisfied a suitable FDT, then the climate response to small external forcing should be calculated by estimating suitable statistics in the
present climate. The climate system is a forced dissipative chaotic dynamical system which is physically quite 
far from the classical physicists' setting for FDT. Leith's suggestion has created a lot of recent activity in generating new
theoretical formulations and approximate algorithms for FDT that were applied to climate response with considerable skill (\cite{GD,GBD,GB,GBM,AM1,AM2,AM3,MW}).

The goal here is to provide a rigorous justification of linear response theory at infinite times for a rather large class of general
stochastic systems, including a class of problems with potential applications to climate change science. 
The improved regularity due to the presence of noise simplifies key aspects of the problem considerably
and circumvents the difficulties found to occur even in relatively simple deterministic systems (see \cite{RuelleDiff,BaladiLR}
for an overview and \cite{MR1233850,MR1679080,MR2201945,MR2433929,MR2399821} for a number of advances in this direction). As a consequence, we are able to focus on
systems that have some other features of realistic systems, namely a lack of ellipticity (we will only require hypoellipticity), non-compactness of state space,
a lack of global Lipschitz continuity of the coefficients, and even infinite-dimensionality of the state space.
The results established  below apply to a very broad class of
nonlinear functionals of state space which trivially include many quantities of interest, such as the mean and variance of subsets of variables. 

The mathematical formulation of FDT as linear response theory for forced dissipative stochastic dynamical systems \cite{Risken,M1,M2,Palmer,MW}
is an appropriate setting for these applications since many current improvements in the comprehensive computer models
for climate change involve stochastic components \cite{PalWil}, while lower dimensional reduced models typically also involve stochastic noise terms
\cite{M2,MAG,MGY}.
While there is a recent formal systematic discussion of FDT for time-dependent stochastic systems of interest in climate change science \cite{MW},
it is not obvious whether the results of this article generalise in a straightforward way.
The remainder of this article is structured as follows. After introducing our setup and notations, a general abstract formulation
of the main Theorem is presented in Section~\ref{sec:main} with a general discussion on its applicability provided in Section~\ref{sec:domain}.
Section~\ref{sec:SDE} illustrates the scope of the theorem for finite-dimensional forced dissipative stochastic systems
with a prototype structure which is relevant for some applications in climate change science.

\subsection*{Acknowledgements}

{\small
We would like to thank the referees for providing us with several literature hints regarding the corresponding problem for deterministic dynamical systems,
as well as for pointing out a number of typographical mistakes in the original draft. 
The research of MH was supported by 
an EPSRC Advanced Research Fellowship (grant number EP/D071593/1) and a Royal Society Wolfson Research Merit Award.
}

\section{Main result}
\label{sec:main}

Throughout the entire article, we will consider a family of Markov evolutions $\{\CP_t^a\,:\, t\ge 0,\, a \in \R\}$ on a separable (either finite
or infinite-dimensional) Banach space $\CB$. In other words, the $\CP_t^a$ are Markov operators over $\CB$ such that 
$\CP_0^a\phi = \phi$ and $\CP_t^a\CP_s^a\phi = \CP_{s+t}^a\phi$ for every bounded measurable function  $\phi \colon \CB \to \R$.
In this abstract setting, we take the pragmatic view that a Markov operator $\CP$ is determined by a transition kernel $P$ as
\begin{equ}
\bigl(\CP \phi\bigr)(x) = \int_\CB \phi(y)\,P(x,dy)\;,
\end{equ}
where the map $A \mapsto P(x,A)$ yields a Borel probability measure for every $x \in \CB$ and that the map $x \mapsto P(x,A)$ is
measurable for every Borel set $A$.
We will denote by $\CP^*$ the adjoint of $\CP$ acting on the space of finite signed Borel measures on $\CB$ by
\begin{equ}
\bigl(\CP^* \mu\bigr)(A) = \int_\CB P(y,A)\,\mu(dy)\;.
\end{equ}

We will also fix a particular value $a_0 \in \R$ and we are interested in the long-time behaviour of the system described by $\CP_t^a$ for $a$ close to $a_0$.
Throughout this article, we assume that for every $a$, there exists a probability measure $\mu_a$ on $\CB$ which is
invariant for $\CP_t^a$, so that 
$\bigl(\CP_t^a\bigr)^*\mu_a = \mu_a$ for every $t \ge 0$.
The aim of this note is to show that, under very weak natural assumptions on the dynamic, the map $a \mapsto \scal{\phi,\mu_a}$ is differentiable at $a_0$
for every sufficiently regular test function $\phi \colon \CB \to \R$. 
(Note that the fact that $a \in \R$ is not a restriction since if we take $a \in \R^n$ instead, it suffices to consider
its components one by one, keeping the others fixed.)

One of the main features of our proof will be that it relies on rather `soft' arguments and has conditions that are easily verifiable in practice.
For example, it does \textit{not} require the transition probabilities to have a density with respect to some reference measure (typically Lebesgue in 
finite dimensions or Gaussian in infinite dimensions).

Denote by $\CC_0^\infty(\CB)$ the set of all functions $\phi\colon \CB \to \R$ such that there exists $N>0$,
a linear map $T\colon \CB \to \R^N$, and a smooth, compactly supported map $\hat \phi\colon \R^N \to \R$ such that
$\phi = \hat \phi \circ T$.
Given two continuous functions $V,W\colon \CB \to [1,\infty)$, we set $\CC^1_{V,W}$ to be the closure of $\CC_0^\infty(\CB)$ under the norm
\begin{equ}[e:normV]
\|\phi\|_{1;V,W} = \sup_{x \in \CB} \Bigl({|\phi(x)| \over V(x)} + {\|D\phi(x)\| \over W(x)}\Bigr)\;.
\end{equ}
If we quotient this space by the space of all constant functions, then this norm is equivalent to the Lipschitz norm 
corresponding to the distance function
\begin{equ}[e:distfcn]
d_{V,W}(x,y) =  \rho_{W}(x,y) \wedge \bigl(V(x) + V(y)\bigr)\;.
\end{equ}
Here, for any positive continuous function $W$ bounded away from $0$, we define the metric
\begin{equ}
 \rho_{W}(x,y) = \inf_{\gamma(0)=x; \gamma(1) = y} \int_0^1 W\bigl(\gamma(s)\bigr)\,\|\dot\gamma(s)\|\, ds\;,
\end{equ}
where the infimum runs over all differentiable curves $\gamma\colon[0,1]\to \CB$ with the prescribed boundary conditions.
(Note that if $W \equiv 1$, then one simply has $\rho_W(x,y) = \|x-y\|$.)
We then make the following assumption:

\begin{assumption}\label{ass:gap}
There exists a time $t>0$ and a constant $\rho < 1$ such that
$\|\CP_t^{a_0} \phi - \scal{\phi,\mu_{a_0}}\|_{1;V,W} \le \rho \|\phi - \scal{\phi,\mu_{a_0}}\|_{1;V,W}$ for every $\phi \in \CC^1_{V,W}$.
\end{assumption}

\begin{remark}
This spectral gap condition is quite different from the $L^2$ or $L^p$ spectral gap conditions often encountered in the literature.
It is however possible to verify that such a condition is satisfied for a very large class of finite-dimensional and infinite-dimensional
(hypo-)elliptic diffusions. See \cite{SGNS} for the verification of Assumption~\ref{ass:gap} in the case of the 2D stochastic Navier-Stokes
equations and Section~\ref{sec:SDE} below for a class of finite-dimensional SDEs.
\end{remark}

In particular, Assumption~\ref{ass:gap} implies that the invariant measure for $\CP_t^{a_0}$ is unique, that the 
spectrum of $\CP_t^{a_0}$ as an operator on $\CC^1_{V,W}$ has an isolated eigenvalue of multiplicity $1$ at $1$,
and that the remainder of its spectrum is contained in the ball of radius $\rho$ around the origin. Therefore, for every centred function
$\phi \in \CC^1_{V,W}$, there exists a unique centred function $\psi \in \CC^1_{V,W}$ such that $\psi - \CP_t^{a_0} \psi = \phi$. We will 
henceforth use the notation $\psi = (1-\CP_t^{a_0})^{-1}\phi$. 

Our second assumption concerns the regularity of $\CP_t^a$ with respect to the parameter $a$. In order to state this assumption, we introduce the
space $\CC_U$, which is the weighted space of continuous functions obtained by completing $\CC_0^\infty(\CB)$ under the norm
\begin{equ}
\|\phi\|_U = \sup_{x \in \CB} {|\phi(x)| \over U(x)} \;.
\end{equ}
We then assume that

\begin{assumption}\label{ass:diff}
There exists a continuous function $U \ge V$ from $\CB$ to $[1,\infty)$ such that, for 
some fixed $t > 0$ and every $\phi \in \CC^1_{V,W}$, the map $a \mapsto \CP_t^a \phi$ is differentiable in a neighbourhood of $a_0$ 
when viewed as a map from $\R$ to $\CC_U$. Denoting this derivative by $\d \CP_t^a \phi$, we furthermore assume that
\begin{equ}[e:boundderPt]
\| \d \CP_t^a \phi\|_U \le C \|\phi\|_{1;V,W}\;,
\end{equ}
for some constant $C>0$.
\end{assumption}

\begin{remark}
Note how we are allowed to lose one derivative as well as some growth in the norm of $\phi$ when computing $\d \CP_t^a \phi$. This is what makes the
verification of \eref{e:boundderPt} very easy in practice, provided that $U$, $V$, and $W$ are chosen in a judicious way, see Section~\ref{sec:SDE} below.
When considering such problems on spaces of probability measures instead of spaces of test functions, it has been know for some time 
that in the context of dynamical systems, spaces of 
$\CC^1$ functions were suitable to take advantage of the expanding character of a map \cite{MR1016871}. Several highly non-trivial extensions of this fact
were considered much more recently in \cite{MR2201945,MR2285731,MR2399821}. See also \cite{MR2264836} for an instance of the use of weights in
a related dynamical systems context.
\end{remark}

Finally, we assume that we have an \textit{a priori} bound on the integrability of the invariant measures:

\begin{assumption}\label{ass:bound}
There exists $\eps > 0$
such that $\sup_{|a - a_0| < \eps} \int_\CB U(x)\,\mu_a(dx) < \infty$.
\end{assumption}

Our main result then states that:

\begin{theorem}\label{theo:linresp}
Let $\{\CP_t^a\}_{a\in \R}$ be a family of Markov semigroups over a separable Banach space $\CB$ such that there exist $\CC^1$ functions $U,V,W\colon \CB \to \R_+$
such that
Assumptions~\ref{ass:gap}, \ref{ass:diff} and \ref{ass:bound} are satisfied for some $a = a_0$ and some fixed $t > 0$. Then the map $a \mapsto \scal{\phi, \mu_a}$ is differentiable 
at $a_0$ for every $\phi \in \CC^1_{V,W}$ and the identity
\begin{equ}[e:linearResponse]
{d\over da} \scal{\phi, \mu_a}\Big|_{a = a_0} = \scal[b]{\d \CP_t^{a_0} (1-\CP_t^{a_0})^{-1}\bigl(\phi - \scal{\phi,\mu_{a_0}}\bigr), \mu_{a_0}}\;,
\end{equ}
holds. In particular, the right hand side of \eref{e:linearResponse} is well-defined.
\end{theorem}

\begin{remark}
As pointed out by one of the referees, the spectral gap condition is not essential for Theorem~\ref{theo:linresp} to hold. It could
be replaced by the requirement that the map $a \mapsto \mu_a$ is continuous in the topology of weak convergence 
and that the test function $\phi$ is such that there exists a function $\psi \in \CC^1_{V,W}$ such that 
$\bigl(1 - \CP_t^{a_0}\bigr) \psi = \phi -  \scal{\phi,\mu_{a_0}}$.

It is not clear how easy the latter condition would prove to verify in practice for systems that do not possess a spectral gap
as it would require a detailed understanding of the
properties of the resolvent which is bypassed by Assumption~\ref{ass:gap}. 
\end{remark}

\begin{remark}
The formula \eref{e:linearResponse} is identical to that familiar to dynamicists, see for example \cite[Eq.~3]{BaladiLR}.
Since the left hand side in \eref{e:linearResponse} is independent of $t$,
we can take the limit $t \to 0$ in the right hand side.
If the test function $\phi$ is centred with respect to $\mu_{a_0}$ and 
sufficiently `nice' so that $t (1-\CP_t^{a_0})^{-1}\phi \to \CL_{a_0}^{-1}\phi$, we then obtain the linear response formula
\begin{equ}[e:LR]
{d\over da} \scal{\phi, \mu_a}\Big|_{a = a_0} = \scal[b]{\d \CL_{a_0}\, \CL_{a_0}^{-1}\phi, \mu_{a_0}}\;,
\end{equ}
where $\d \CL_a$ is the formal derivative of the generator $\CL_a$ with respect to the 
parameter $a$. (This can also be obtained by formally differentiating the identity $\CL_a^* \mu_a = 0$ with
respect to the parameter $a$.) 

If the state space is $\R^n$, $\mu_{a_0}$ has a density $\exp(-H(x))$ 
with respect to Lebesgue measure
and $\d \CL_{a_0} = F(x)\nabla_{\!x}$ is a first-order differential operator, then we can define the conjugate
current $J$ by
\begin{equ}
J(x) = F(x) \nabla_{\!x} H(x) - \div_x F(x)\;.
\end{equ}
With this notation, \eref{e:LR} can be interpreted as the Green-Kubo relations
\begin{equ}
{d\over da} \scal{\phi, \mu_a}\Big|_{a = a_0} = \scal[b]{J \CL_{a_0}^{-1}\phi, \mu_{a_0}} = \int_0^\infty \E \bigl(J(x_0) \phi(x_t)\bigr)\,dt\;,
\end{equ}
where expectations are taken over the stationary Markov process $x_t$ with invariant measure $\mu_{a_0}$.
\end{remark}

\begin{remark}
Note that even though we assume uniqueness of the invariant measure for $\CP_t^{a_0}$, this need not be the case for $a \neq a_0$.
Even though uniqueness doesn't need to hold for $a \neq a_0$, our result implies differentiability of the map $a \mapsto \scal{\phi, \mu_a}$ at $a = a_0$ 
for any choice of $a \mapsto \mu_a$. This is analogous to the result shown for partially hyperbolic systems in \cite{MR2031432}.
\end{remark}

\begin{remark}
The classical perturbation results found for example in \cite{Kato}
do not seem to be applicable directly to our situation since we do \textit{not} assume that the map $a \mapsto \CP_t^a \phi$ is differentiable in $\CC^1_{V,W}$.
\end{remark}

\begin{proof}
Let us first show that the map $a \mapsto \mu_a$ is Lipschitz continuous at $a = a_0$ if we endow the space of probability measures
on $\CB$ which integrate $d_{V,W}(0,x)$ with the norm dual to that of $\CC_{V,W}^1$. For every $\phi \in \CC_{V,W}^1$, we have the identity
\begin{equ}
\scal{\phi, \mu_a - \mu_{a_0}} = \scal{\CP_t^{a_0} \phi, \mu_a - \mu_{a_0}} + \scal{(\CP_t^a - \CP_t^{a_0})\phi, \mu_{a}}\;.
\end{equ}
Since none of these expressions changes if we add a constant to $\phi$, we can assume in the sequel that $\phi$ is centred with respect
to $\mu_{a_0}$ so that, using the invertibility of $(1-\CP_t^{a_0})$ and setting $\psi = (1 - \CP_t^{a_0}) \phi$, we obtain the identity
\begin{equ}[e:exprdiff]
\scal{\psi, \mu_a - \mu_{a_0}} = \scal{(\CP_t^a - \CP_t^{a_0})(1-\CP_t^{a_0})^{-1}\psi, \mu_{a}}\;.
\end{equ}
It now follows readily from Assumptions~\ref{ass:diff} and \ref{ass:bound} that 
\begin{equs}
\sup_{\|\phi\|_{1;V,W} = 1} \scal{\phi, \mu_a - \mu_{a_0}} &= \sup_{\|\phi\|_{1;V,W} = 1} \scal{(\CP_t^a - \CP_t^{a_0})(1-\CP_t^{a_0})^{-1}\phi, \mu_{a}}\\
&\le C|a-a_0| \sup_{\|\psi\|_{U} = 1} \scal{\psi, \mu_{a}} \le \tilde C|a-a_0|\;,
\end{equs}
for some positive constants $C$ and $\tilde C$, as required. 

Since the norm dual to that of $\CC_{V,W}^1$ (restricted to differences between probability measures) 
is nothing but the Wasserstein-$1$ distance corresponding to $d_{V,W}$ \cite{Villani}, it follows that $\mu_a \to \mu_{a_0}$ weakly for
$a \to a_0$. Combining this with Assumption~\ref{ass:bound}, we conclude in particular that 
\begin{equ}[e:convW]
\scal{\Psi,\mu_a} \to \scal{\Psi,\mu_{a_0}}\;,\qquad  \forall \Psi \in \CC_U\;. 
\end{equ}
(This follows by a standard `diagonal argument' from the fact that $\Psi$ can be approximated by continuous 
bounded functions by the definition of $\CC_U$.)

We now have all the preliminaries in place to show that $\scal{\phi,\mu_a}$ is differentiable and to compute its derivative. 
It follows from Assumption~\ref{ass:diff} that, for every (centred) $\phi \in \CC^1_{V,W}$, 
there exists a function $a \mapsto \CR_a \in \CC_U$ with $\lim_{a \to a_0} \|\CR_a\|_U = 0$ and such that we have the identity
\begin{equ}
{1\over a-a_0} (\CP_t^a - \CP_t^{a_0})(1-\CP_t^{a_0})^{-1}\phi = \d \CP_t^{a_0} (1-\CP_t^{a_0})^{-1}\phi + \CR_a \;.
\end{equ}
As a consequence of \eref{e:exprdiff}, we have
\begin{equs}
\scal{\phi, \mu_a - \mu_{a_0}} &= \scal{(\CP_t^a - \CP_t^{a_0})(1-\CP_t^{a_0})^{-1}\phi, \mu_{a_0}} \\
&\quad + \scal{(\CP_t^a - \CP_t^{a_0})(1-\CP_t^{a_0})^{-1}\phi, \mu_{a} - \mu_{a_0}}\;,
\end{equs}
so that 
\begin{equs}
\Bigl|{\scal{\phi, \mu_a - \mu_{a_0}} \over a-a_0} -  \scal{\d \CP_t^{a_0} (1-\CP_t^{a_0})^{-1}\phi, \mu_{a_0}}\Bigr|
&\le 2|\scal{\CR_a, \mu_{a_0}}| + |\scal{\CR_a, \mu_{a}}| \\
&+ |\scal{\d \CP_t^{a_0} (1-\CP_t^{a_0})^{-1}\phi, \mu_a - \mu_{a_0}}|\;.
\end{equs}
The first two terms on the right hand side converge to zero thanks to the continuity of $\CR_a$ at $a_0$ combined with
Assumption~\ref{ass:bound}, while the last term converges to zero thanks to \eref{e:convW}, thus concluding the proof.
\end{proof}

\section{Domain of applicability}
\label{sec:domain}

In this section, we discuss under which circumstances one can expect a spectral gap result in a norm of the type \eref{e:normV}.
Harris' theorem \cite{Harris,MT} provides a general methodology to verify whether a system satisfies a spectral gap in $\CC_V$ for some $V$.
If it so happens that such a system also has the property that one can find $W$ such that,
for some fixed $t_0>0$, $\CP_{t_0}^{a_0}$ is a bounded operator from $\CC_V$ into $\CC_{V,W}^1$, then it follows that $\CP_t^{a_0}$ has a spectral
gap in $\CC_{V,W}^1$. However, such an approach presents two problems:
\begin{enumerate}
\item While Harris' theorem is suitable for finite-dimensional diffusions, there are many examples of stochastic PDEs satisfying a spectral gap
in a space of the type $\CC_{V,W}^1$ for which Harris' theorem does not apply.
\item In finite dimensions, as soon as a system satisfies H\"ormander's bracket condition, $\CP_t^{a_0}$ maps continuous functions into smooth functions.
In particular, there always exists some $W$ such that it maps $\CC_V$ into $\CC_{V,W}^1$. The problem is that there does not appear to be a general 
result available that controls the growth of such a function $W$. This is problematic, since the function $U$ appearing in Assumption~\ref{ass:diff}
is generally comparable to $W$, so that Assumption~\ref{ass:bound} is hard to verify without control on $W$.
\end{enumerate}

Our aim is therefore to present conditions under which a spectral gap in  a space of the type $\CC_{V,W}^1$ can be verified directly, without
requiring a spectral gap result in a space of the type $\CC_V$.
In \cite{WeakHarris}, a general framework was laid out that gives conditions under which a ``spectral gap'' is obtained in a Wasserstein-type
distance. Unfortunately, the Wasserstein distance in question is with respect to a metric comparable to the square root of the original
metric on the space. Therefore, in our setting, this actually yields a spectral gap in a space of ${1\over 2}$-H\"older continuous functions, which
is useless for our purpose. On the other hand, a spectral gap in a space of the type \eref{e:normV} was obtained for the stochastic Navier-Stokes
equations in \cite{SGNS}.

In this section, we obtain a slightly different set of conditions under which a spectral gap in a norm of the type \eref{e:normV}
can be shown for a general Markov operator $\CP$ over a Banach space $\CB$. 
Our conditions are inspired by those in \cite{SGNS,WeakHarris}, aiming in particular at being applicable to 
equations with conservative quadratic nonlinearities, as arising in applications.
The main ingredient is again a gradient bound of the following type, which was shown to hold for a large class of
stochastic PDEs in \cite{NS,Hypo}:

\begin{assumption}\label{ass:contr}
There exist continuous functions $U_1, U_2\colon \CB \to \R_+$ such that, for every $\eps>0$ there is a constant $C_\eps > 0$ such that the bound
\begin{equ}[e:gradientbound]
\|D\CP \phi(x)\|^2 \le \eps  U_1^2(x) \bigl(\CP \|D\phi\|^2\bigr)(x) + C_\eps U_2^2(x) \bigl(\CP \phi^2\bigr)(x)\;,
\end{equ}
holds for every $x \in \CB$ and every $\phi \in \CC^1_b(\CB)$.
\end{assumption}

\begin{remark}
In the case of the stochastic 2D Navier-Stokes equations, this bound can be shown to hold with $U_1(x) = U_2(x) = \exp (\eta \|x\|^2)$ for arbitrarily small $\eta > 0$, 
with $x$ the vorticity of the vector field and $\|x\|$ its $L^2$-norm.
\end{remark}

The aim of this section is to show that if a Markov operator $\CP$ satisfies Assumption~\ref{ass:contr} and has sufficiently good contraction properties,
then it is possible to deduce a spectral gap in the norm \eref{e:normV} (for a suitable choice of $V$ and $W$), provided that the transition probabilities
satisfy a kind of topological irreducibility condition. This statement can be formulated precisely in the following way:

\begin{theorem}\label{theo:HarrisWeak}
  Let $\CP$ be a Markov operator over a separable Banach space $\CB$ mapping $\CC_0^\infty(\CB)$ into
  $\CC^1_{V,W}$ and satisfying Assumption~\ref{ass:contr}.  Suppose furthermore that there
  exist continuous functions $V,W\colon \CB \to \R_+$ satisfying
\begin{equ}[e:superlyap]
U_1^2 \CP W^2 + U_2^2 \CP V^2 \le C W^2\;, \qquad \CP V \le {1\over 2} V + K\;,
\end{equ}
for some constants $C$ and $K$.

Finally, assume that there exists a point $x_\star \in \CB$ such that, 
 for every $\eps > 0$ and every $C>0$ there exists $\alpha > 0$ such that 
\begin{equ}[e:irred]
\inf_{x\,:\,V(x) \le C} \CP\bigl(x,B_\eps(x_\star)\bigr) \ge \alpha\;.
\end{equ}

Then, $\CP$ has exactly one invariant probability measure $\mu_\star$. Furthermore, there exist constants $C$ and $\gamma > 0$ such that the bound
\begin{equ}[e:gap]
\|\CP^n \phi - \scal{\phi,\mu_\star}\|_{1;V,W} \le C e^{-\gamma n} \|\phi - \scal{\phi,\mu_\star}\|_{1;V,W} \;,
\end{equ} 
holds for every $\phi \in \CC^1_{V,W}$. Finally, there exist constants $\delta > 0$, $\beta > 0$ and $\rho <1$
such that the bound
\begin{equ}[e:gaptime1]
\|\CP \phi - \scal{\phi,\mu_\star}\|_{1;1+\beta V,\delta^{-1}W} \le \rho \|\phi - \scal{\phi,\mu_\star}\|_{1;1+\beta V,\delta^{-1}W} \;,
\end{equ} 
holds for every $\phi \in \CC^1_{V,W}$. 
\end{theorem}

\begin{remark}
It will follow from the proof that it is sufficient that the second inequality in \eref{e:superlyap} holds with 
${1\over 2}$ replaced by any other constant $\kappa < 1$.
Furthermore, we do not need \eref{e:irred} to hold for every $C>0$, 
but only for some $C > K/(1-\kappa)$.
\end{remark}

\begin{proof}
Since the norms $\|\,\cdot\,\|_{1;1+\beta V,\delta^{-1}W}$ and $\|\,\cdot\,\|_{1;V,W}$
are equivalent for every choice of $\beta, \delta > 0$, \eref{e:gap} follows immediately from 
\eref{e:gaptime1}, so that it is sufficient to check \eref{e:gaptime1}.

The proof uses very similar ideas to the proofs of \cite[Theorem~4.7]{WeakHarris} and \cite[Theorem~3.4]{SGNS}.
In particular, we use the following trick. For $\delta > 0$ and $\beta > 0$,
we introduce the distance 
\begin{equ}
\hat d_{\beta,\delta}(x,y) = \delta^{-1}\rho_W(x,y) \wedge \bigl(2 + \beta V(x) + \beta V(y)\bigr)\;.
\end{equ}
This distance is of course equivalent to the distance $d_{V,W}$ introduced in \eref{e:distfcn}, but it turns out that allowing the freedom
of choosing both $\delta$ and $\beta$ sufficiently small will considerably simplify the proofs.
With this definition, a $\CC^1$ function $\phi$ is Lipschitz continuous with Lipschitz constant $1$ with respect to $\hat d_{\beta,\delta}$ if and only if
\begin{equ}[e:lipphi]
\|D\phi(x)\| \le \delta^{-1} W(x)\;,\qquad |\phi(x)| \le 1+ \beta V(x)\;.
\end{equ}
(For the second inequality, one might have to add a suitable constant to $\phi$.) Denote by $\Lip$ the set of all such functions.
As in \cite{WeakHarris,SGNS}, we show that it is possible
to choose $\delta$ and $\beta$ in such a way that the bound
\begin{equ}[e:requestedbound]
\hat d_{\beta,\delta}\bigl(\CP(x,\cdot\,),\CP(y,\cdot\,)\bigr) \le \alpha d_{\beta,\delta}(x,y)\;,
\end{equ}
holds for some $\alpha < 1$ uniformly over all pairs $x,y \in \CB$. We again show \eref{e:requestedbound} separately in three different cases
and we use separately the three ingredients of the theorem in each of these cases.

\textbf{The case $\boldsymbol{\rho_W(x,y) \le \delta \bigl(2 + \beta V(x) + \beta V(y)\bigr)}$.} In this case, we make use of the gradient
bound \eref{e:gradientbound}, together with the `super-Lyapunov' structure \eref{e:superlyap} to deduce that if $\phi$ satisfies
\eref{e:lipphi}, then for every $\eps > 0$ there exists $C_\eps$ such that the bound
\begin{equ}
\|D\CP\phi(x)\| \le \eps \delta^{-1} W(x) + C_\eps W(x)\;,
\end{equ}
holds uniformly for all such $\phi$ and for all $\beta \le 1$, say. It follows that by first choosing $\eps = {1 \over 4}$ and then
choosing $\delta$ small enough so that $C_\eps \le 1/(4\delta)$, one has
\begin{equ}
\|D\CP\phi(x)\| \le {1\over 2\delta}W(x)\;,
\end{equ} 
which immediately implies that
\begin{equs}
\hat d_{\beta,\delta}\bigl(\CP(x,\cdot\,),\CP(y,\cdot\,)\bigr) &\le \sup_{\phi \in \Lip}  |\CP\phi(x) - \CP\phi(y)| \\
&\le \sup_{\phi \in \Lip} \inf_{\gamma}
\int_0^1 \|D\CP\phi(\gamma(s))\|\,|\dot\gamma(s)|\,ds \\
&\le {1\over 2\delta} \sup_{\phi} \inf_{\gamma} \int_0^1 W(\gamma(s))\,|\dot\gamma(s)|\,ds \\
&\le {1\over 2\delta} \rho_W(x,y) \le {1\over 2}\hat d_{\beta,\delta}(x,y)\;,
\end{equs}
as requested.

\textbf{\mathversion{bold}The case $\rho_W(x,y) > \delta \bigl(2 + \beta V(x) + \beta V(y)\bigr)$ and $V(x) + V(y) \ge 4(K+2)$.} In this case,
we only make use of the fact that $V$ is a Lyapunov function. We have indeed
\begin{equs}
\hat d_{\beta,\delta}\bigl(\CP(x,\cdot\,),\CP(y,\cdot\,)\bigr) &\le  \sup_{\phi \in \Lip}|\CP\phi(x) - \CP\phi(y)| \le  \sup_{\phi \in \Lip}\bigl(|\CP\phi(x)| + |\CP\phi(y)|\bigr) \\
&\le 2 + \beta \CP V(x) + \beta \CP V(y) \le 2 + {\beta\over 2} \bigl(V(x) + V(y)\bigr) + \beta K\\
&\le (2-2\beta) + {3\beta\over 4} \bigl(V(x) + V(y)\bigr) \le \bigl(1-(\beta\wedge \textstyle{1\over 4})\bigr)\hat d_{\beta,\delta}(x,y) \;,
\end{equs}
which again yields a contraction, but with a strength that depends this time on the parameter $\beta$. Finally, we have

\textbf{\mathversion{bold}The case $\rho_W(x,y) > \delta \bigl(2 + \beta V(x) + \beta V(y)\bigr)$ and $V(x) + V(y) < 4(K+2)$.} 
In this case, we make use of our final assumption, namely \eref{e:irred}. At this stage we assume that $\delta > 0$ is fixed, sufficiently small
so that our first step goes through. We can then find some sufficiently small $\eps > 0$ so that $\hat d_{\beta,\delta}(x_\star, y) \le {1\over 2}$
for all $y \in B_\eps(x_\star)$, uniformly over $\beta \le 1$. In this case, we can decompose the Markov operator $\CP$ into
a combination $\CP = \alpha\CP_1 + (1-\alpha)\CP_2$ of Markov operators such that $\CP_1\bigl(x,B_\eps(x_\star)\bigr) = 1$
for every $x$ such that $V(x) \le 4(K+2)$. We conclude that
\begin{equs}
\hat d_{\beta,\delta}\bigl(\CP(x,\cdot\,),\CP(y,\cdot\,)\bigr) &\le \alpha \hat d_{\beta,\delta}\bigl(\CP_1(x,\cdot\,),\CP_1(y,\cdot\,)\bigr) + (1-\alpha)\hat d_{\beta,\delta}\bigl(\CP_2(x,\cdot\,),\CP_2(y,\cdot\,)\bigr)\\
&\le \alpha + (1-\alpha)\bigl(2 + \beta \CP_2 V(x) + \beta \CP_2 V(x)\bigr)\\
&\le \alpha + 2(1-\alpha) + \beta \CP V(x) + \beta \CP V(x) \\
&\le 2-\alpha + {\beta\over 2} \bigl(V(x) + V(y) + 2K\bigr) \le 2-\alpha + \beta (3K + 4)\;.
\end{equs}
We can now choose $\beta$ sufficiently small so that this constant is strictly smaller than $2$. Since on the other hand one has $\hat d_{\beta,\delta}(x,y) \ge 2$,
the claim now follows.
\end{proof}

\section{Application to stochastic differential equations}
\label{sec:SDE}

In this section, we aim to apply the results of the previous sections to a concrete class of stochastic
differential equations. In order to keep our conditions clean, we restrict ourselves to SDEs with additive noise and
polynomial nonlinearities. We thus assume that for some integers $N,M \ge 1$, we have
\begin{equ}[e:SDE]
dx(t) = \sum_{k=0}^N \N_k(x,\ldots,x)\,dt + \sum_{k=1}^M \sigma_k\,dw_k(t)\;,\qquad x(0) \in \R^d\;.
\end{equ}
Here, the functions $\N_k$ are symmetric multilinear maps of order $k$ and the $\sigma_k$ are elements of $\R^d$. We
also define the function $\N\colon \R^d \to \R^d$ by
\begin{equ}
\N(x) = \sum_{k=0}^N \N_k(x,\ldots,x)\;,
\end{equ}
as well as the $d\times M$ matrix $\Sigma$ given by $\Sigma = (\sigma_1,\ldots,\sigma_M)$.

When $N=2$, equations of this type typically arise as effective dynamic for large-scale structures in climate models
\cite{M1}. Fixing $N$ and $M$, we make the following structural assumptions on the $\N_k$'s and the $\sigma_k$'s:

\begin{assumption}\label{ass:coercive}
There exist constants $c,C>0$ such that the bound
\begin{equ}[e:bounddissip]
\scal{x,\N(x)} \le C - c \|x\|^{N}\;,\quad \forall x\in\R^d
\end{equ}
holds.
\end{assumption}

\begin{remark}
The bound \eref{e:bounddissip} can hold even if the degree $N$ of the nonlinearity is even. In this
case, it implies that $\scal{x,\N_N(x,\ldots,x)} = 0$. If $N$ is odd, then \eref{e:bounddissip} would actually
also hold with $N$ replaced by $N+1$, but this is not important for our purpose.
\end{remark}

\begin{assumption}\label{ass:hypo}
Define $A_k \subset \R^d$ recursively by $A_0 = \span\{\sigma_1,\ldots,\sigma_M\}$ and
\begin{equ}
A_{k+1} = \span \bigl( A_k \cup \{\N_N(y_1,\ldots,y_N)\,:\, y_j \in A_k\}\bigr)\;.
\end{equ}
Then, one has $A_d = \R^d$.
\end{assumption}

\begin{assumption}\label{ass:control}
The control system associated to \eref{e:SDE} is approximately controllable. In other words, for every
initial condition $x_0$, every $T>0$, every terminal condition $x_T$, and every $\eps > 0$, there 
exists a smooth control $u \in \CC^\infty([0,T],\R^M)$ such that the solution to
\begin{equ}
\dot x(t) = \N(x(t)) + \sigma\, u(t)\;,\qquad x(0) = x_0\;,
\end{equ}
satisfies $\|x(T) - x_T\| \le \eps$.
\end{assumption}

\begin{remark}
If $N$ is odd, then Assumption~\ref{ass:control} is redundant since it then follows from
Assumption~\ref{ass:hypo}, see \cite{PolyOdd}. However, in the more interesting
case where $N$ is even, the situation is not so easy. A number of criteria are given in
\cite{Sussmann,Kawski,PolyGen}, but they are all sufficient without being necessary. 
In the particular case of truncations of the two and three-dimensional Navier-Stokes
equations where furthermore $\Sigma$ is assumed to be diagonal in the Fourier basis, 
it was also shown that Assumption~\ref{ass:control} follows from
Assumption~\ref{ass:hypo}  \cite{MatE,Romito,AgrSar}.
\end{remark}

For given $M,N>0$, let $\J$ denote the set of $(\N,\Sigma)$ such that Assumptions~\ref{ass:coercive},
\ref{ass:hypo}, and \ref{ass:control} hold. Note that $(\N,\Sigma)$ can be viewed as a vector in $\R^n$ with
$n = d \bigl(M +\sum_{k=0}^N \bigl({d+k \atop k}\bigr)\bigr)$, so that $\J$ is nothing but some possibly rather complicated
open subset of $\R^n$.
With these assumptions in hand, we have the following result, which follows immediately
from \cite{DPZ,MT}:

\begin{theorem}
For every $(\N,\Sigma) \in \J$, the system \eref{e:SDE} possesses a unique 
invariant probability measure  $\mu_{(\N,\Sigma)}$.
\end{theorem}

\begin{proof}
The existence of an invariant measure follows from Assumption~\ref{ass:coercive} by Krylov-Bogoliubov.
The uniqueness is then a consequence of the strong Feller property (which follows from Assumption~\ref{ass:hypo})
and the topological irreducibility (which follows from Assumption~\ref{ass:control}) 
of the associated Markov semigroup.
\end{proof}

The aim of this section is to prove the following result:

\begin{theorem}\label{theo:main}
Fix $M,N\ge 1$ and let $\phi \colon \R^d\to \R$ be a $\CC^1$ function such that
\begin{equ}
\sup_{x \in \R^d} \bigl(|\phi(x)| + \|D\phi(x)\|\bigr)e^{-\delta \|x\|^2} < \infty\;,
\end{equ}
for every $\delta > 0$. Then, the map $(\N,\Sigma) \mapsto \int_{\R^d} \phi(x)\,\mu_{(\N,\Sigma)}(dx)$ is
$\CC^1$ on $\J$ and its derivative is given by the formula \eref{e:linearResponse}.
\end{theorem}

Before we turn to the proof of Theorem~\ref{theo:main}, we start by providing the 
\textit{a priori} bounds on the solutions to \eref{e:SDE} required for our analysis. For this, we first introduce the
Jacobian process of \eref{e:SDE}, which is the $d\times d$ matrix-valued two-parameter process $J_{s,t}$ given by 
the solution to the following random ODE:
\begin{equ}[e:defJac]
{dJ_{s,t} \over dt} = D\N(x)J_{s,t}\;,\qquad J_{s,s} = \mathrm{Id}\;.
\end{equ}
We also introduce the second variation process $J^{(2)}_{s,t}$ which is a bilinear map-valued process given by
\begin{equ}
{dJ^{(2)}_{s,t}(\xi,\zeta) \over dt} = D\N(x)J^{(2)}_{s,t}(\xi,\zeta) + D^2\N(x)\bigl(J_{s,t}\xi, J_{s,t}\zeta\bigr) \;,\qquad J^{(2)}_{s,s} = 0\;,
\end{equ}
for every $\xi$ and $\zeta$ in $\R^d$. We then have the following:

\begin{proposition}\label{prop:bounds}
Let $(\N,\Sigma)$ be such that Assumption~\ref{ass:coercive} holds. Then, there exist strictly positive constants $\eta$, $\nu$ and $C$ and, for every $\eps > 0$ and $p>0$ there exists a constant $C_{p,\eps}$ such that the bounds
\begin{equs}
\E \exp\Bigl(\eta \|x(t)\|^2 + \eta \int_0^t \|x(s)\|^N\,ds \Bigr) &\le C\exp \bigl(\eta e^{-\nu t} \|x(0)\|^2\bigr)\;, 
\label{e:boundSol}\\
\E \bigl(\|J_{s,t}\|^p + \|J^{(2)}_{s,t}\|^p \bigr) &\le C_{p,\eps} \exp \bigl(\eps \|x(0)\|^2\bigr)\;,
\label{e:boundJac}
\end{equs}
hold for every initial condition $x(0) \in \R^d$, and for every $s,t \in [0,1]$.
\end{proposition}

\begin{proof}
Applying It\^o's formula to the function $x \mapsto \|x\|^2$, it follows from our assumption that
\begin{equ}
{1\over 2} d\|x(t)\|^2 \le - c \|x\|^N\,dt + C\,dt + \sum_{k=1}^M \scal{x(t),\sigma_k}\,dw_k(t)\;,
\end{equ}
for a possibly different constant $C$.
Note that this inequality should really be interpreted in integral form. The bound \eref{e:boundSol} 
then follows  from the version of the exponential martingale inequality given in \cite[Lemma~5.1]{SGNS}.

In order to obtain the bound \eref{e:boundJac},
note that it follows from \eref{e:defJac} that there exists a constant $C$ such that 
\begin{equ}
\|J_{s,t}\|^p \le \exp\Bigl(pC \int_s^t \bigl(1 + \|x(r)\|^{N-1}\bigr)\,dr\Bigr)\;.
\end{equ}
The requested bound now follows immediately from \eref{e:boundSol}, noting that for every $\delta>0$ 
there exists a constant $\tilde C_{p,\delta}$
such that the inequality
\begin{equ}
pC\bigl(1 + \|x\|^{N-1}\bigr) \le \delta \|x\|^N + \tilde C_{p,\delta}\;,
\end{equ}
holds for every $x \in \R^d$.
\end{proof}

The main ingredient for all probabilistic proofs of regularising properties for SDEs is the \textit{Malliavin matrix}
of the process. This is the $\R^{d\times d}$-valued matrix $\CM_t$ defined by
\begin{equ}
\CM_t = \int_0^t J_{s,t} \Sigma \Sigma^* J_{s,t}^*\,ds\;,
\end{equ}
where $J_{s,t}$ is the Jacobian of the flow as above and $\Sigma$ is the noise matrix defined at the
start of this section. Note that $\CM_t$ is a random object which furthermore depends on the initial
condition $x_0$ of \eref{e:SDE}. We have the following \textit{a priori} bound:

\begin{proposition}\label{prop:Malliavin}
Let $(\N,\Sigma)$ be such that Assumptions~\ref{ass:coercive} and \ref{ass:hypo} hold.
Then, for every $t>0$, every $\delta > 0$, and every $p\ge 1$ there exists a constant $C>0$ such that the bound
\begin{equ}
\E \|\CM_t^{-{1\over 2}}\|^p \le C \exp\bigl(\delta \|x_0\|^2\bigr)\;,
\end{equ}
holds for every $x_0 \in \R^d$.
\end{proposition}

\begin{proof}
This is an immediate consequence of \cite[Thm~6.7]{Hypo}. The assumptions of that theorem are satisfied
by Proposition~\ref{prop:bounds} 
if we choose $\psi_0(x) = \exp(\eta \|x\|^2)$ for a sufficiently small constant $\eta > 0$. 
\end{proof}

As a consequence of Proposition~\ref{prop:Malliavin}, we obtain the following gradient bound
for the Markov semigroup $\CP_t$ generated by solutions to \eref{e:SDE}:

\begin{proposition}\label{prop:gradient}
Let $(\N,\Sigma)$ be such that Assumptions~\ref{ass:coercive} and \ref{ass:hypo} hold.
Then, for every $t>0$ and every $\eta > 0$ there exists constant $C$ such that the bound
\begin{equ}
\| D\CP_t \phi(x)\|^2 \le C \exp\bigl(\eta \|x\|^2\bigr) \bigl(\CP_t \phi^2\bigr)(x)\;,
\end{equ}
holds for every $x \in \R^d$ and every $\phi \in \CC_b^1(\R^d)$.
\end{proposition}

\begin{proof}
The proof is almost identical to that of \cite[Thm~5.4]{Hypo}, which is in turn
largely inspired by that of \cite[Thm~3.2]{Norris}. We therefore only give a sketch of the proof
and refer to \cite{Norris,Hypo,Nualart} for details.
Let $\xi \in \R^d$ be an arbitrary deterministic
vector with $\|\xi\| = 1$ and fix a terminal time $t>0$. We then 
define the $\R^M$-valued stochastic process $h_\xi$ by
\begin{equ}
h_\xi(s) = \Sigma^* J_{s,t}^* \CM_t^{-1} J_{0,t}\xi\;.
\end{equ}
Note that $t$ is a \textit{fixed} terminal time, so that $h_\xi$ is a stochastic process
which is not adapted to the filtration generated by the Wiener processes driving \eref{e:SDE}.
With this notation at hand, it follows from the integration by parts formula on Wiener space
\cite{Nualart} that one has the identity
\begin{equ}
\scal{D\CP_t \phi(x),\xi} = \E\Bigl(\phi(x_t)\int_0^t h_\xi(s)\,dW(s)\Bigr)\;,
\end{equ}
where the stochastic integral should be interpreted as a Skorokhod integral \cite{Nualart}.
It then follows from the generalised It\^o isometry \cite{Nualart} that
\begin{equ}
\scal{D\CP_t \phi(x),\xi}^2 \le \bigl(\CP_t \phi^2\bigr)(x)\,\E\Bigl(\int_0^t \|h_\xi(s)\|^2\,ds + \int_0^t \|\D_r h_\xi(s)\|^2 \,ds\,dr \Bigr)\;,
\end{equ}
where $\D_r F$ denotes the Malliavin derivative at time $r$ of the random variable $F$. It follows from
the chain rule that one has the identity
\begin{equs}
\D_r h_\xi(s) &= \Sigma^* \D_r J_{s,t}^* \CM_t^{-1} J_{0,t}\xi + \Sigma^* J_{s,t}^* \CM_t^{-1} \D_r J_{0,t}\xi \\
& \quad - \int_0^t \Sigma^* J_{s,t}^* \CM_t^{-1} \bigl(\D_r J_{u,t} \Sigma \Sigma^* J_{u,t}^* + J_{u,t} \Sigma \Sigma^* \D_r J_{u,t}^*\bigr)\CM_t^{-1}J_{0,t}\xi\,du \;.
\end{equs}
Note now that the Malliavin derivative of the Jacobian $J_{s,t}$ is given by
\begin{equ}
\D_r J_{s,t}\xi = 
\left\{\begin{array}{cl}
	0 & \text{if $r \ge t$,} \\
	J_{s,t}^{(2)}(J_{r,s} \Sigma , \xi)& \text{if $r \le s$,} \\
	J_{r,t}^{(2)}(\Sigma, J_{s,r}\xi) & \text{otherwise.}
\end{array}\right.
\end{equ}
Collecting all of these identities, it then follows immediately from the a priori bounds of Proposition~\ref{prop:bounds}
that the requested bound holds.
\end{proof}

With all of these preliminary results at hand, we can now proceed to the proof of Theorem~\ref{theo:main}:

\begin{proof}[of Theorem~\ref{theo:main}]
As a first step, we show that the assumptions of Theorem~\ref{theo:HarrisWeak} hold.
It follows from Proposition~\ref{prop:gradient} that Assumption 4 is satisfied with $U_1(x) = 0$ and
$U_2(x) = \exp\bigl(\delta \|x\|^2\bigr)$ for arbitrary choice of $\delta > 0$. In view of
Proposition~\ref{prop:bounds}, it is therefore natural to choose $V(x) = W(x) = \exp(\eta \|x\|^2)$
for some sufficiently small value of $\eta > 0$. It is possible to first choose $\eta$ small enough so that the bounds
in Proposition~\ref{prop:bounds} hold and then $\delta$ small enough so that the bounds \eref{e:superlyap} hold.

We now want to show that \eref{e:irred} holds as well. It follows from Assumption~\ref{ass:control} combined
with the support theorem \cite{StroockVaradhan} that $\CP_t(x,B_\eps(y)) > 0$ for every $x,y \in \R^d$, $t > 0$
and $\eps > 0$. Since furthermore the Markov semigroup $\CP_t$ is Feller and the level sets of $V$ are compact,
\eref{e:irred} follows at once. We thus conclude that, for every $t>0$, the assumptions of 
Theorem~\ref{theo:HarrisWeak} hold, so that 
the semigroup $\CP_t$ possesses a spectral gap in 
the space $\CC^1_{V,V}$ for sufficiently small values $\eta > 0$. 

We now show that Theorem~\ref{theo:linresp} can be applied to our situation, if we choose $a$ to be any
of the components of either $\N$ or $\Sigma$.
First of all, it follows from our previous discussion that Assumption~\ref{ass:gap} is satisfied
for every $t>0$, provided that we choose $V(x) = 1+\beta \exp(\eta \|x\|^2)$ and $W(x) = \delta^{-1} \exp(\eta \|x\|^2)$
for sufficiently small values $\beta$, $\delta$, and $\eta$ (possibly depending on the value of $t$).
Furthermore, if we set $U(x) = 2\exp(\eta' \|x\|^2)$ for some sufficiently small value $\eta' > \eta$, then it follows
from \eref{e:boundSol} that Assumption~\ref{ass:bound} is also satisfied.
It therefore suffices to check that Assumption~\ref{ass:diff} is also satisfied.

Denote by $\Phi_t(x,\N,\Sigma)$ the solution to \eref{e:SDE} with initial condition $x$ and parameters $\N$
and $\Sigma$. By considering $\N$ and $\Sigma$ as additional dynamical variables that follow the evolution
equations $\dot \N = 0$, $\dot \Sigma = 0$, it
 follows from \cite{Kunita,ElwCarv} that on $\R^d \times \J$,
$\Phi_t$ is a $\CC^\infty$ map of all of its arguments. Furthermore, it follows from the variation of constants
formula that its derivative with respect to $\N$ and $\Sigma$
in arbitrary directions $\delta\N$ and $\delta \Sigma$ is given by
\begin{equ}
\scal{D_\N \Phi_t, \delta \N} = \int_0^t J_{s,t} \delta \N(x_s)\,ds\;,\quad 
\scal{D_\Sigma \Phi_t, \delta \Sigma} = \int_0^t J_{s,t} \delta \Sigma\,dW(s)\;.
\end{equ}
Combining this with the bounds from Proposition~\ref{prop:bounds}, it follows at once that Assumption~\ref{ass:diff}
is satisfied with
\begin{equ}
V(x) = 1+\beta e^{\eta \|x\|^2}\;,\quad W(x) = \delta^{-1}e^{\eta \|x\|^2}\;,\quad U(x) = V(x) + e^{2\eta \|x\|^2}\;,
\end{equ}
provided that $\beta$, $\delta$, and $\eta$ are sufficiently small. This allows us to
apply Theorem~\ref{theo:linresp}, thus concluding the proof.
\end{proof}


\label{sec:SPDEs}

\bibliographystyle{Martin}
\bibliography{./Majda}

\end{document}